\documentclass[conference]{IEEEtran}

\usepackage{amssymb}
\usepackage{amsthm}
\usepackage{algorithm}
\usepackage{algorithmic}
\usepackage{amsmath}
\allowdisplaybreaks[1]
\usepackage{amsfonts}
\usepackage{graphicx}
\usepackage{makeidx}
\usepackage{cite}
\usepackage{float}
\usepackage{epstopdf}
\usepackage{subfigure}
\usepackage{array}
\usepackage{hyperref}

\newtheorem{theorem}{\textbf{Theorem}}

\newtheorem{lemma}{\textbf{Lemma}}

\newtheorem{remark}{\textbf{Remark}}

\IEEEoverridecommandlockouts

\graphicspath{{figures/}}
\DeclareMathSizes{10}{10}{7}{3}
\let\sss = \scriptscriptstyle
\begin{document}
%

\title{Optimal Hybrid Full-Duplex/Half-Duplex Scheme for Buffer Aided Relay Systems}

\author{
\IEEEauthorblockN{Cheng~Li\IEEEauthorrefmark{1},~Bin~Xia\IEEEauthorrefmark{1},~Pihe~Hu\IEEEauthorrefmark{1}，~Yao~Yao\IEEEauthorrefmark{3}}

\IEEEauthorblockA{\IEEEauthorrefmark{1}Department of Electronic Engineering, Shanghai Jiao Tong University, Shanghai, China}


\IEEEauthorblockA{\IEEEauthorrefmark{3}Huawei Technologies Co. Ltd}

Emails: \{lichengg, bxia, hupihe\}@sjtu.edu.cn, yyao@eee.hku.hk
}

\maketitle
\begin{abstract}
Full-duplex (FD) communication has received great interest in recent years due to the potential of doubling the spectral efficiency. However, how to alleviate the detrimental effects of the residual self-interference (RSI) incurred by the FD mode is still a challenging problem. In this paper, focusing on the statistical throughput maximization, we propose an optimal hybrid FD/half-duplex (HD) scheme for the one-way FD buffer aided relay system. To solve this problem, we divide the system into four different transmission modes and formulate the problem as a binary integer programming problem. By relaxing the binary variables to be continuous ones, we solve the problem using the Karush-Kuhn-Tucker (KKT) optimal conditions. We obtain the selection probability of each mode based on the instantaneous channel outage states. The proposed scheme not only achieves the optimal FD or HD mode selection, but also realizes adaptive source-to-relay or relay-to-destination link selection. Simulation results show that the proposed scheme offers 95\% maximum gain over the HD counterparts.
\end{abstract}


%

\section{Introduction}
Deploying relay nodes in networks is an efficient way to improve the quality of service of cell edge users\cite{6146495}. Field measurements have shown that relay nodes can improve the system throughput and enhance the outdoor to indoor coverage\cite{6666597}. Owing to the various advantages, relay communications have received great interest. For instance, in \cite{5557651}, the outage performance was studied for the bidirectional half-duplex (HD) decode-and-forward (DF) relay systems.
Recently, the full-duplex (FD) mode, which receives and transmits in the same frequency band simultaneously, was proposed \cite{6832464}.
In \cite{7410116}, the authors showed that the FD mode could efficiently improve the performance of the relay systems over the HD counterparts.
However, the performance gain was still greatly limited by the residual self-interference (RSI) after self-interference cancellation. In \cite{7568989}, the authors analyzed the outage performance of the two-way FD relay system with RSI. The results revealed that the HD counterparts outperformed the FD system in the high RSI regime.

To overcome the detrimental effects of the RSI, the authors in \cite{5961159} have proposed a hybrid HD/FD scheme based on the instantaneous transmission rate. However, the authors did not consider the buffer at the relay node. The received packets needed to be forwarded immediately. Hence, the end-to-end performance was limited by the worse one of the $S$-$R$ and $R$-$D$ links. In \cite{6330084}, the authors have proposed an adaptive link selection protocol for the traditional HD relay system with a buffer.
The results revealed that the buffer could efficiently improve the end-to-end throughput.
Intuitively, to deploy FD mode in the aforementioned systems will improve the system performance as it bears the potential to double the spectral efficiency. However, we cannot simply extend the previous schemes to the FD cases due to the simultaneous transmission and reception. In addition, the RSI, which is a critical factor that limits the system performance, should be carefully treated as well. The relaying scheme and the hybrid FD/HD mode need to be re-designed to exploit the maximum potential of the FD mode. To the best of the authors' knowledge, few works has been done on this topic.

To utilize the advantages of the buffer, in \cite{7248607,8094980}, the authors have designed the hybrid FD/HD protocols to improve the system throughput. In these protocol the FD mode is always preferred whenever it is available. In fact, we find that the FD mode is not always preferred. In addition, in \cite{8094980}, the HD and FD mode may be both possible. In this case, we design the selection probabilities to select different modes. In this paper, we are dedicated to design a hybrid FD/HD protocol to maximize the system statistic throughput over the infinite time horizon for the one-way FD DF relay system. To solve this problem, we first divide the system into four different transmission modes and then formulate the problem as a binary integer programming problem. By relaxing the binary variables to be continuous ones, we solve the problem using the Karush-Kuhn-Tucker (KKT) optimal conditions. The proposed scheme not only achieves the optimal FD or HD mode selection but also realizes adaptive $S$-$R$ or $R$-$D$ link selection. In addition, the general expressions of the system throughput is derived. Simulation results validate that the proposed scheme effectively overcomes the detrimental effects of the RSI and offers 95\% maximum throughput gain over the HD counterparts.




\section{System Model}
In this section, we describe the system model, including the channel mode, basic transmission modes and different channel gain regions.
\begin{figure}[ht]
  \centering
  \includegraphics[width=2.8in]{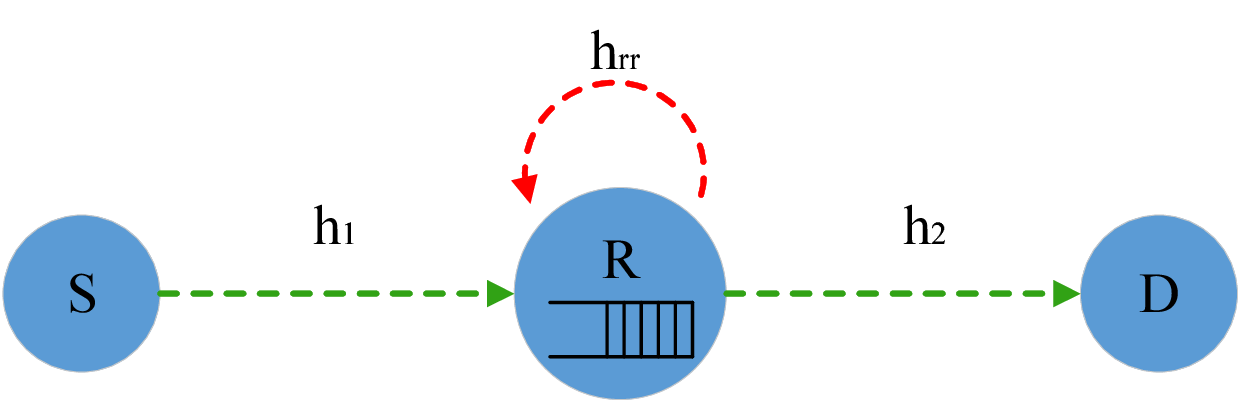}\\
  \vspace{-3mm}
  \caption{The three-node one-way FD DF relay system model.}
  \label{system_model}
\end{figure}

\subsection{Channel Model}
The system model is depicted in Fig. 1. The direct link between $S$ and $D$ does not exist. The relay node has the capability to work in the full-duplex (FD) mode and is equipped with a buffer $B$. In addition, we denote the channel coefficients of the $S$-$R$ and $R$-$D$ links as $h_{1}(i)$ and $h_{2}(i)$ in the $i$-th time slot, respectively. The block fading channels are considered. We assume that $|h_{1}(i)|$ and $|h_{2}(i)|$ are subject to the independent, stationary and ergodic random distributions. All the channels are contaminated by the thermal noise. When the relay node works in the FD mode, the $S$ to $R$ link will be impaired by the RSI as well after self-interference cancellation. We use $\gamma_{1}^{\sss F}(i)$, $\gamma_{1}^{\sss H}(i)$ and $\gamma_{2}^{}(i)$ to denote the signal-to-interference-plus-noise ratios (SINRs) of the $S$-$R$ link in the FD mode, HD mode and the $R$-$D$ link, respectively, which are given by
\begin{equation}\label{1684179434}
  \gamma_{1}^{\sss F}(i)=\frac{P_{1}g_{1}(i)}{I_{R}+\sigma_{r}^2},\  \gamma_{1}^{\sss H}(i)=\frac{P_{1}g_{1}(i)}{\sigma_{r}^2},\ \gamma_{2}^{}(i)=\frac{P_{2}g_{2}(i)}{\sigma_{d}^2},
\end{equation}
where $P_{1}$ and $P_{2}$ denote the transmit powers of $S$ and $R$, respectively. $g_{1}(i)=|h_{1}(i)|^2$ and $g_{2}(i)=|h_{2}(i)|^2$ denote the channel gains of the $S$-$R$ link and $R$-$D$ link, respectively. $\sigma_{r}^2$ and $\sigma_{d}^2$ denote the variances of the noise at $S$ and $D$, respectively. $I_{R}=K_{R}P_{2}$, where $K_{R}$ is the indicator of the RSI level. In this paper, we consider the fixed rate transmission scenario with rate $R_{0}$, which has been considered in various applications, such as packet transmission. In addition, we assume that the relay node knows the global CSI.

\subsection{Different Transmission Modes}
In the system, there exists four different transmission modes, denoted by $M_{1}$, $M_{2}$, $M_{3}$, $M_{4}$. To be specific, in the mode $M_{1}$, the source $S$  sends messages to the relay $R$, but $R$ keeps silent and the messages are stored in the buffer $B$. In the mode $M_{2}$, the source $S$ keeps silent, but the relay $R$ extracts messages from the buffer $B$ and forwards them to the destination $D$. In the mode $M_{3}$, the system works in the FD mode, i.e., $S$ and $R$ transmit signals simultaneously.

In order to specify which mode is selected, we let $d_{j}(i) \in \{0,1\}$ to serve as the mode indicator. For instance, when the mode $M_{j}$ is selected in the \emph{i-th} time slot, we set $d_{j}(i)$ to $1$, otherwise we set $d_{j}(i)$ to $0$. In each time slot, only one mode can be selected, we obtain that $\sum_{j=1}^4d_{j}(i)=1$.
In order to decide whether the information can be successfully decoded in the mode $M_{j}$, we adopt $O_{j}(i) \in \{0, 1\}$, to serve as the transmission indicator. If the transmission in the mode $M_{j}$ is successful, we set $O_{j}(i)$ to 1, otherwise we set $O_{j}(i)$ to 0. For the mode $M_{4}$, since all the nodes keep silent, we always set $O_{4}=0$.

\subsection{Different Channel Gain Regions}
\begin{figure}[ht]
  \centering
  \includegraphics[width=2.8in]{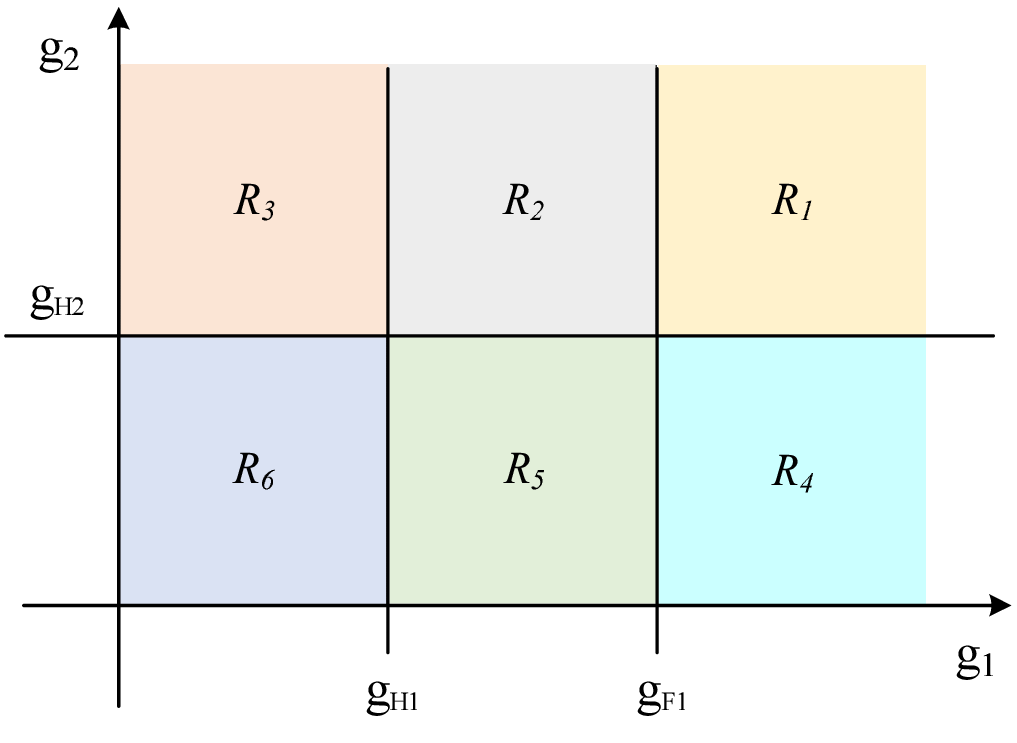}\\
  \vspace{-3mm}
  \caption{Different channel gain regions.}
  \label{system_model}
\end{figure}

For better illustration, we divided the whole channel gain plain into six un-overlapped regions, which are denoted by $\mathcal{R}_{k}, k\in\{1,...,6\}$ as shown in the Fig. 2, where $g_{1}^{\sss H}=\frac{\gamma_{0}\sigma_{r}^2}{P_{1}}$ and $g_{1}^{\sss F}=\frac{\gamma_{0}(I_{\sss R}+\sigma_{r}^2)}{P_{1}}$ denote the outage threshold of the $S$-$R$ link under the HD mode and FD mode, respectively. $g_{2}^{\sss H}=\frac{\gamma_{0}\sigma_{d}^2}{P_{2}}$ denotes the outage threshold of the $R$-$D$ link.
For clarity, the relationships between the transmission indicators and the different channel gain regions are summarized in the following table.
\begin{table}[ht]
\begin{center}
\caption{Viability of Different Modes}
\par
\begin{tabular}{|c|c|c|c|c|c|c|}
  \hline
  Channel Gain Regions & $\mathcal{R}_{1}$ & $\mathcal{R}_{2}$ &$\mathcal{R}_{3}$&$\mathcal{R}_{4}$&$\mathcal{R}_{5}$&$\mathcal{R}_{6}$\\ \hline
  $O_{1}$ & 1 & 1&0&1&1&0 \\ \hline
  $O_{2}$ & 1 & 1 &1&0&0&0\\ \hline
  $O_{3}$&1&0&0&0&0&0 \\
  \hline
\end{tabular}
\end{center}
\end{table}


\vspace{-0mm}
\section{Optimal Hybrid FD/HD Mode Selection Scheme}
In this section, we first formulate the system statistical throughput maximization problem. Then, the optimal mode selection scheme is presented.

\subsection{Problem Formulation}
In this paper, we desire to devise an optimal hybrid FD/HD scheme to maximize the long-term system throughput. First, the rate of the $S$-$R$ link is given by
\begin{equation}\label{74875418}
  R_{1} = \lim_{N\rightarrow +\infty}\frac{1}{N} \sum_{i=1}^{N}[d_{1}(i)O_{2}(i)+d_{3}(i)O_{3}(i)]\times R_{0}.
\end{equation}

Similarly, the maximum average rate of the $R$-$D$ link is given by \begin{equation}\label{47868}
R_{2}=\lim_{N\rightarrow +\infty}\frac{1}{N} \sum_{i=1}^{N}[d_{2}(i)O_{2}(i)+d_{3}(i)O_{3}(i)]\times R_{0}.
\end{equation}

In this paper, we consider a stationary buffer state at the relay node. Hence, we have the following lemma.
\begin{lemma}
For a stationary stochastic system, the average arrival rate needs to be smaller than the departure rate, i.e.,
\begin{equation}\label{82740}
  R_{1}\leq R_{2}.
\end{equation}
\end{lemma}
\begin{proof}
According to the queuing theory, if $R_{1}>R_{2}$, the queue in the buffer will statistically increase to infinity. Hence, we obtain this lemma.
\end{proof}

\begin{lemma}
 The optimal hybrid FD/HD scheme maximizing the system throughput is achieved when the queue at the buffer is at the edge of non-absorb state, which can be expressed as
 \begin{equation}\label{247845348}
  R_{1} = R_{2}.
\end{equation}
\begin{proof}
According to the \emph{Lemma 1}, $R_{1}\leq R_{2}$. If $R_{1} < R_{2}$, the system throughput will be limited by $R_{1}$. However, we can choose the mode $M_{1}$ in more time slots to increase the rate $R_{1}$. Once $R_{1}=R_{2}$, $R_{1}$ will cannot be increased furthermore, otherwise, the system cannot keep stationary.
\end{proof}
\end{lemma}

The system throughput can be quantified by the average received bits of information at the destination node $D$ in each time slot. Hence, we have
\begin{equation}\label{787984}
  \mathcal{T}=R_{2}.
\end{equation}

\begin{table*}[t]
\newcommand{\tabincell}[2]{\begin{tabular}{@{}#1@{}}#2\end{tabular}}
\begin{center}
\caption{Viability of Different Modes with Different SINR regions}
\par
\begin{tabular}{|c|c|c|c|c|}
  \hline
   & $\mathcal{R}_{1}$ & $\mathcal{R}_{2}$ &$\mathcal{R}_{3}$&$R_{4}\cup \mathcal{R}_{5}$\\ \hline
  $\Psi_{1}$ &$P_{1}^2=1$  &$P_{2}^2=1$&$P_{3}^2=1$&\tabincell{c}{\!\!\!$P_{4}^1=\frac{P_{\mathcal{R}_{1}}+P_{\mathcal{R}_{2}}+P_{\mathcal{R}_{3}}}{P_{\mathcal{R}_{4}}+P_{\mathcal{R}_{5}}}$\!\!\!\\$P_{4}^4=1-P_{4}^1$} \\ \hline
  $\Psi_{2}$ & \tabincell{c}{\!\!\!$P_{1}^2=\frac{P_{\mathcal{R}_{4}}+P_{\mathcal{R}_{5}}-P_{\mathcal{R}_{2}}-P_{\mathcal{R}_{3}}}{P_{\mathcal{R}_{1}}}$\!\!\!\\$P_{1}^3=1-P_{1}^2$}&$P_{2}^2=1$ &$P_{3}^2=1$&$P_{4}^1=1$\\ \hline
  $\Psi_{3}$&$P_{1}^3=1$&\tabincell{c}{\!\!\!$P_{2}^2=\frac{P_{\mathcal{R}_{2}}+P_{\mathcal{R}_{4}}+P_{\mathcal{R}_{5}}-P_{\mathcal{R}_{3}}}{2P_{\mathcal{R}_{2}}}$\!\!\!\\$P_{2}^3=1-P_{2}^2$}&$P_{3}^2=1$&$P_{4}^1=1$ \\ \hline
  $\Psi_{4}$&\tabincell{c}{\!\!\!$P_{1}^1=\frac{P_{\mathcal{R}_{3}}-P_{\mathcal{R}_{2}}-P_{\mathcal{R}_{4}}-\mathcal{R}_{5}}{P_{\mathcal{R}_{1}}}$ \!\!\! \\ $P_{1}^3=1-P_{1}^1$}&$P_{2}^1=1$&$P_{3}^1=1$&$P_{4}^1=1$ \\ \hline
  $\Psi_{5}$&$P_{1}^1=1$ &$P_{2}^1=1$ &\tabincell{c}{\!\!\!$P_{3}^2=\frac{P_{\mathcal{R}_{1}}+P_{\mathcal{R}_{2}}+P_{\mathcal{R}_{4}}+\mathcal{R}_{5}}{P_{\mathcal{R}_{3}}}$\!\!\!\\$P_{3}^4=1-P_{3}^2$}&$P_{4}^1=1$ \\ \hline
\end{tabular}
\end{center}
\end{table*}

It is noted that we abandon the limitation $\min\{Q(i-1),R_{0}\}$. This can be interpreted as that there are only countable number of time slots that the queue length of the buffer is less than $R_{0}$ when the buffer is at the edge of non-absorb state. In other words, the relay node always has enough data stored in the buffer. Hence, when averaged on the infinite time horizon, the effects of $Q(i)<R_{0}$ can be neglected.



Now, based on \emph{Lemma 1} and \emph{Lemma 2}, the considered throughput maximization problem can be formulated as $\mathcal{P}_{1}$, which is given by
\begin{align}\label{473928}
  \mathop {\text{max}} \limits_{d_{j}(i)\in\{0,1\}} \ & \quad\ \mathcal{T} \notag\\
  \text{s.t.\quad C1:} &\quad R_{1}=R_{2},  \notag\\
   \text{C2:} &\quad \sum_{j=1}^{4} d_{j}(i)=1,\ \forall i\notag\\
   \text{C3:} &\quad  d_{j}(i) \in \{0, 1\},\ \forall i, k
\end{align}

However, since the variable $d_{j}(i)$ can only be $1$ or $0$, the original optimization problem $\mathcal{P}_{1}$ is a binary integer programming problem, which is hard to solve. Alternatively, we observe that if we relax the binary variable $d_{j}(i)$ to be continuous one in the interval $[0,1]$, the original problem will transfer to a linear programming problem. For the linear programming problem, the optimal value always obtained at the vertexes of the feasible set, i.e., $d_{j}(0) =1\ \text{or}\ 0$, which is consistent with the original problem \cite{7809043}. The standard form of the relaxed problem $\mathcal{P}_{2}$ is given by
\begin{align}\label{4394574}
  \mathop {\text{min}} \limits_{d_{j}(i)\in[0,1]} \ & \quad\ -\mathcal{T} \notag\\
  \text{s.t.\quad C1:} &\quad R_{1}-R_{2}=0,  \notag\\
   \text{C2:} &\quad \sum_{j=1}^{4} d_{j}(i)-1=0,\ \forall\ i\notag\\
   \text{C3:} &\quad  d_{j}(i)-1\leq0,\ \forall\ i, k\notag\\
   \text{C4:} &\quad -d_{j}(i)\leq 0,\ \forall\ i,k
\end{align}



\subsection{Optimal Hybrid FD/HD Scheme}
Due to that the optimal scheme relies on the statistic CSI, we first present the following different statistic CSI cases
\begin{align}\label{45654654}
  \Psi_{1}: &\quad P_{\mathcal{R}_{3}}\leq P_{\mathcal{R}_{4}}+P_{\mathcal{R}_{5}}-P_{\mathcal{R}_{1}}-P_{\mathcal{R}_{2}},\notag\\
  \Psi_{2}: & \quad  P_{\mathcal{R}_{4}}+P_{\mathcal{R}_{5}}-P_{\mathcal{R}_{1}}-P_{\mathcal{R}_{2}}< P_{\mathcal{R}_{3}}\leq P_{\mathcal{R}_{4}}+P_{\mathcal{R}_{5}}-P_{\mathcal{R}_{2}},\notag\\
  \Psi_{3}: & \quad P_{\mathcal{R}_{4}}+P_{\mathcal{R}_{5}}-P_{\mathcal{R}_{2}} < P_{\mathcal{R}_{3}} \leq P_{\mathcal{R}_{4}}+P_{\mathcal{R}_{5}}+P_{\mathcal{R}_{2}},   \notag\\
  \Psi_{4}: & \quad P_{\mathcal{R}_{4}}+P_{\mathcal{R}_{5}}+P_{\mathcal{R}_{2}} < P_{\mathcal{R}_{3}}\leq  P_{\mathcal{R}_{4}}+P_{\mathcal{R}_{5}}+P_{\mathcal{R}_{2}}+P_{\mathcal{R}_{1}},\notag\\
  \Psi_{5}: &\quad P_{\mathcal{R}_{3}}>P_{\mathcal{R}_{4}}+P_{\mathcal{R}_{5}}+P_{\mathcal{R}_{2}}+P_{\mathcal{R}_{1}},
\end{align}
where $P_{\mathcal{R}_{k}}$ denotes the probability of the channel gain region $\mathcal{R}_{k}$. As the relay node has known the statistic information of the channel variation, thus the probability of the regions from $\mathcal{R}_{1}$ to $\mathcal{R}_{6}$ can be easily derived.

\begin{theorem}
  The optimal hybrid FD/HD mode selection scheme maximizing the system throughput of the considered buffer aided FD relay system with RSI is given in the Table II.
\\

\noindent where $P_{k}^j, k\in\{1,2,3,4,5,6\}, j\in\{1,2,3,4\}$ denotes the selection probability of the mode $M_{j}$ in the region $\mathcal{R}_{k}$. In the region $\mathcal{R}_{6}$, all the active modes are unaccessible, only the inactive mode $M_{4}$ can be selected, i.e., $P_{6}^4=1$. It is noted that the selection probabilities, which equals to $0$, are not given.
\begin{proof}
See Appendix.
\end{proof}
\end{theorem}

\begin{figure*}[tbp]
\centering
\subfigure[]{
\label{Fig.sub.1}
\begin{minipage}{2.0in}
\centering
\includegraphics[width=2.2in]{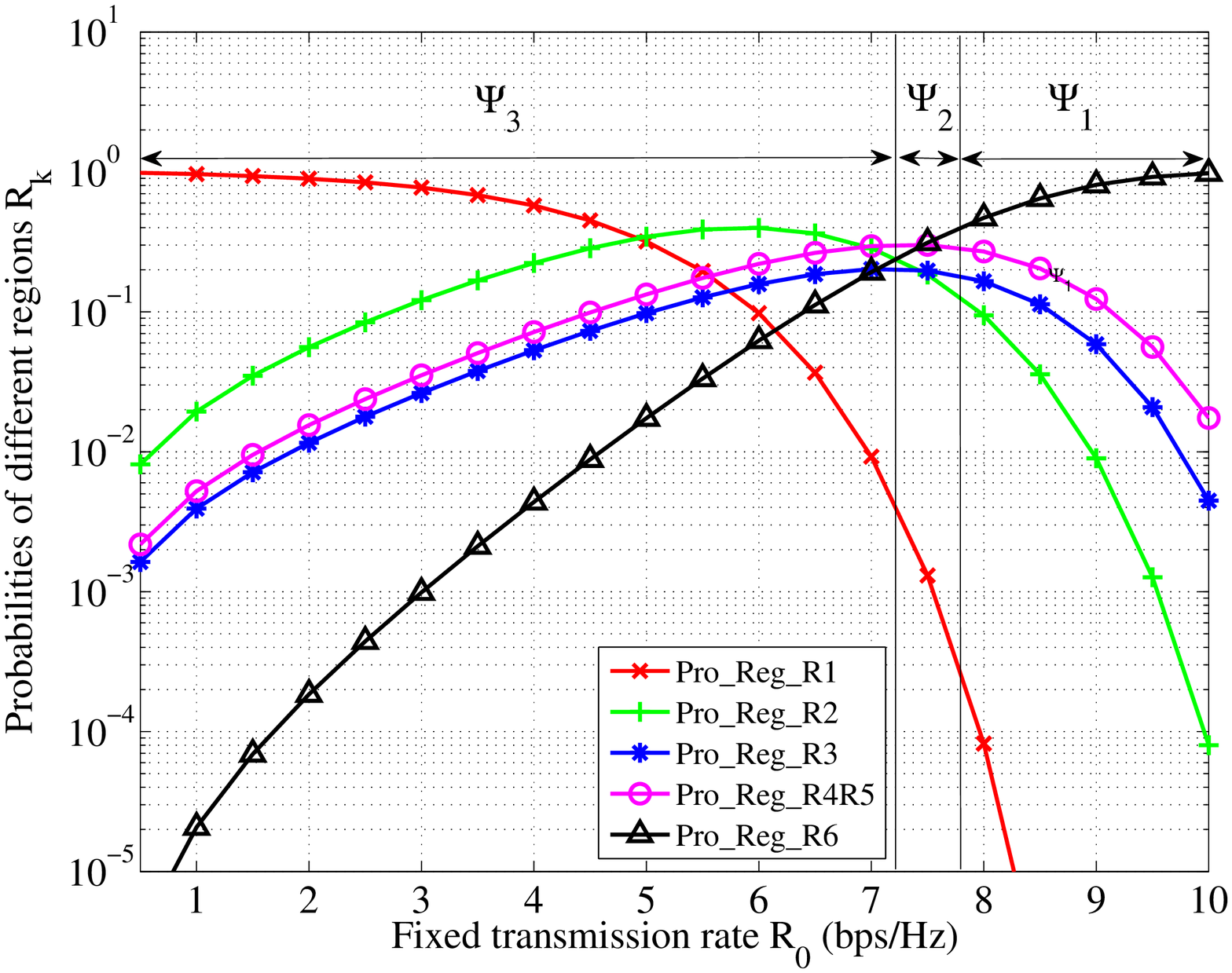}
\end{minipage}
}
\subfigure[]{
\begin{minipage}{2.0in}
\centering
\includegraphics[width=2.2in]{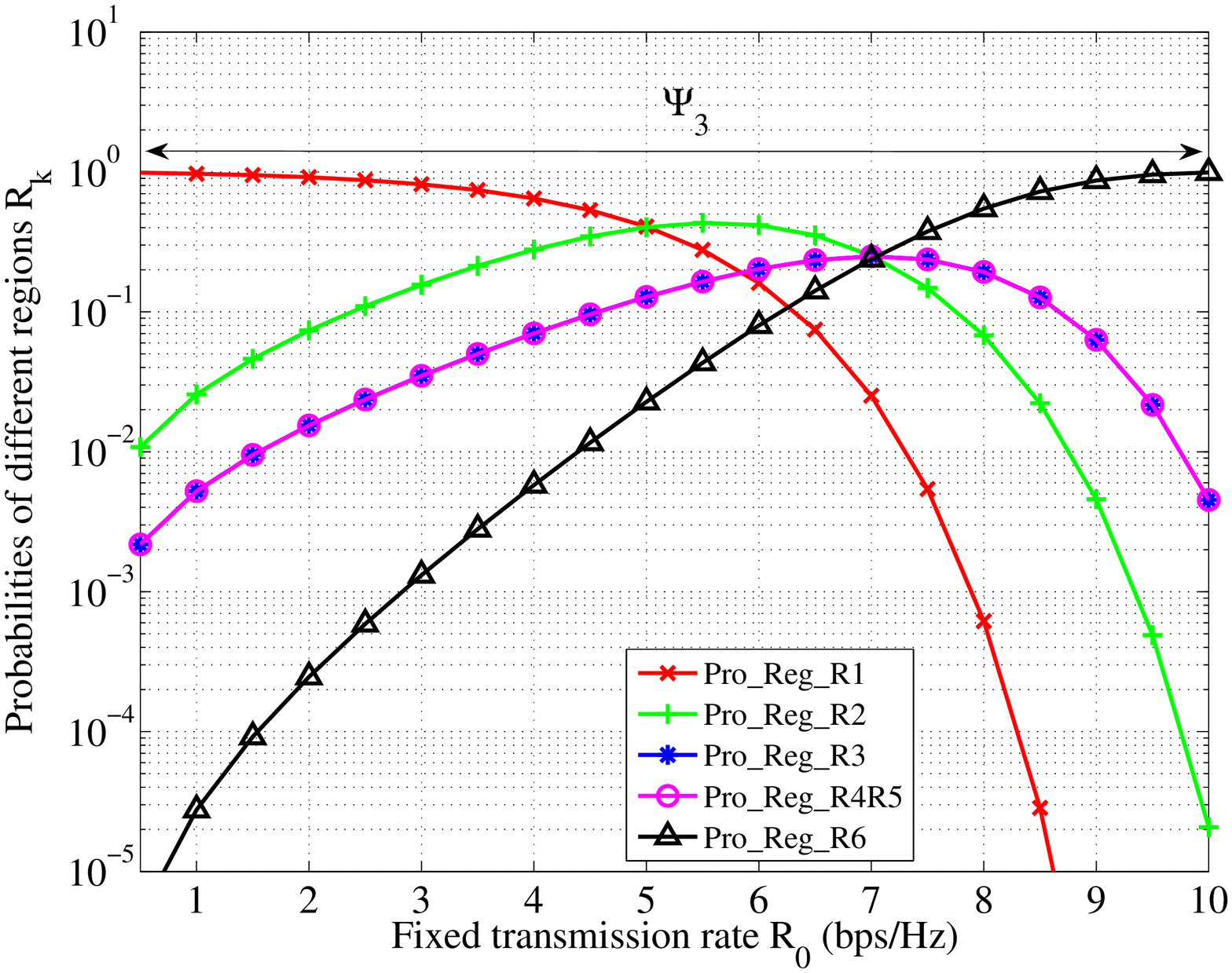}
\end{minipage}
}
\subfigure[]{
\begin{minipage}{2.0in}
\centering
\includegraphics[width=2.2in]{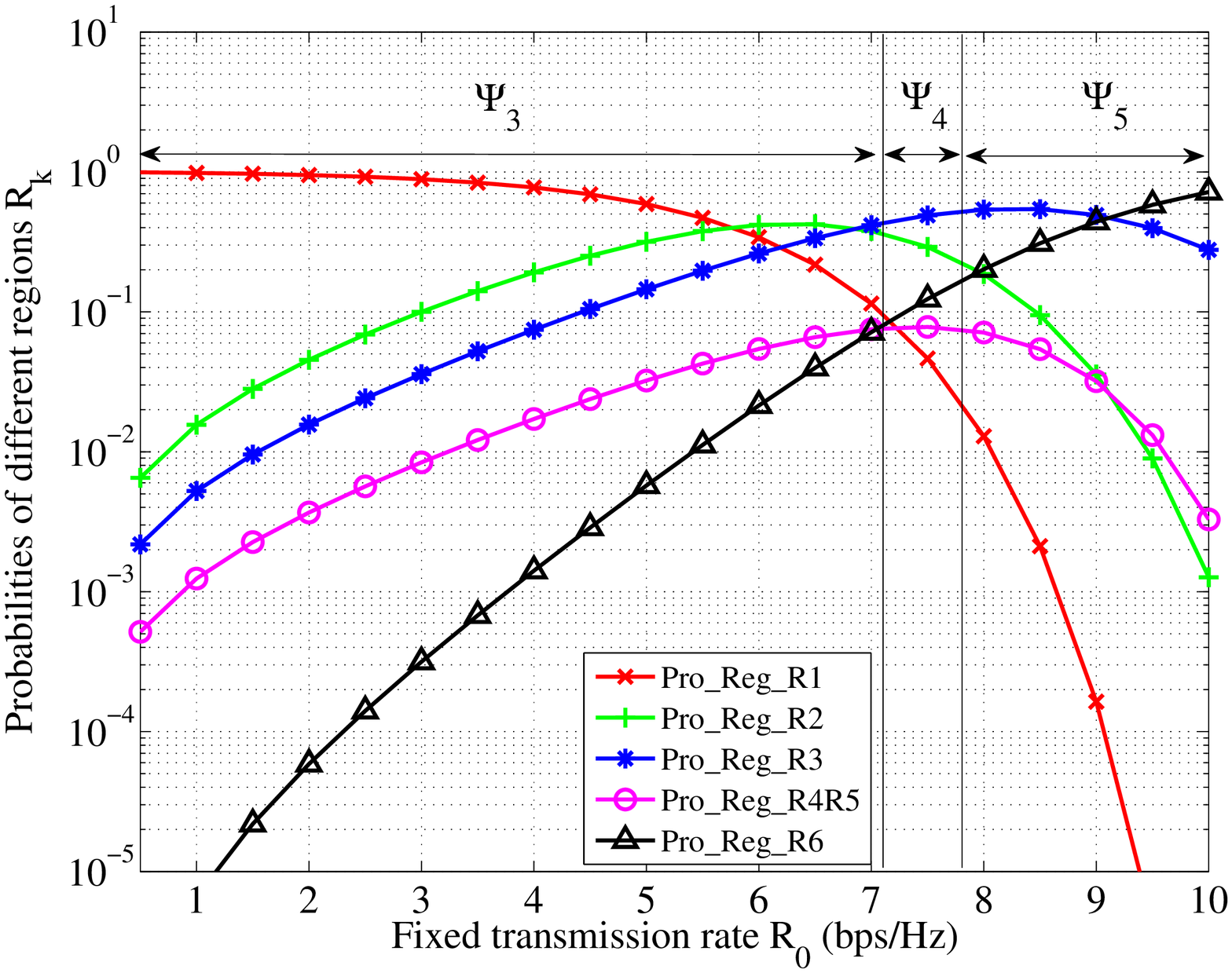}
\end{minipage}
}
\vspace{-3mm}
\caption{Probabilities vs. the fixed transmission rate $R_{0}$ under different system setups, in which 1) $S$-$R$ is stronger than $R$-$D$ link; 2) $S$-$R$ link is equivalently equal to the $R$ to $D$ link; 3) $S$-$R$ link is weaker than $R$-$D$ link.}
\label{Fig.lable}
\end{figure*}

\begin{remark}
The five non-overlapping statistic CSI cases is proposed in the design of the optimal scheme (please see the Appendix). The five cases represent five different relative qualities of the $S$-$R$ and $R$-$D$ links. In the case $\Psi_{1}$, the $S$-$R$ link is far stronger than the $R$-$D$ link. Hence, all the time slots in the region $\mathcal{R}_{1}$ need to be allocated to the mode $M_{2}$ to balance the constraint C1. In the case $\Psi_{2}$, the $S$-$R$ link is slightly stronger than the $R$-$D$ link. We only need to allocate a part of the time slots in the region $\mathcal{R}_{1}$ to the mode $M_{2}$. In the case $\Psi_{3}$, the $S$-$R$ link and the $R$-$D$ link are comparable. By allocating the time slots in the $\mathcal{R}_{2}$ to the mode $M_{1}$ or $M_{2}$ could make the constraint C1 satisfied. Thus, all the time slots in the region $\mathcal{R}_{1}$ could be allocated to the mode $M_{3}$. In the case $\Psi_{4}$ and $\Psi_{5}$, the situation will be contrary, more time slots need to be allocated to the mode $M_{1}$.
\end{remark}

\begin{remark}
\emph{Theorem 1} validates that the FD mode is not always preferred even it is available.  In addition, in the case of $\Psi_{1}$ and $\Psi_{5}$, the FD mode $M_{3}$ has never been selected, which means that the system always works in the pure HD mode. In other words, although the relay node has the capacity to work in the FD mode, the system cannot achieve any performance gain due to the strong disparities of the qualities of the $S$-$R$ link and the $R$-$D$ link.
\end{remark}

\subsection{General Expressions of the System Throughput}
In this subsection, we present the general expressions of the system throughput of the proposed hybrid FD/HD scheme.

\begin{theorem}
  The general expressions of the maximum system throughput  of the proposed hybrid FD/HD scheme corresponding to different statistic CSI cases are given by
  \begin{equation}\label{4214034140}
    \mathcal{T}=\left\{
    \begin{array}{ll}
    \vspace{3mm}
      (P_{\mathcal{R}_{1}}+P_{\mathcal{R}_{2}}+P_{\mathcal{R}_{3}})R_{0}, &\quad\Psi_{1},\ \Psi_{2}\\
    \vspace{3mm}
      (P_{\mathcal{R}_{1}}+\frac{P_{\mathcal{R}_{2}}+P_{\mathcal{R}_{3}}+P_{\mathcal{R}_{4}}+P_{\mathcal{R}_{5}}}{2})R_{0}, &\quad\Psi_{3}\\
       (P_{\mathcal{R}_{1}}+P_{\mathcal{R}_{2}}+P_{\mathcal{R}_{4}}+P_{\mathcal{R}_{5}})R_{0},&\quad\Psi_{4}, \ \Psi_{5}
    \end{array}
    \right.
  \end{equation}
\begin{proof}
  \emph{Theorem 2} can be easily proved using the different selection probabilities shown in the Theorem 1. The system average throughput can be generally obtained by
\begin{equation}\label{13216871}
    \mathcal{T}=[P_{\mathcal{R}_{1}}(P_{1}^2+P_{1}^3)+P_{\mathcal{R}_{2}}P_{2}^2+P_{\mathcal{R}_{3}}P_{3}^2]R_{0}.
\end{equation}

Substituting $P_{j}^k$ in the Theorem 1 into the above expression, the expressions of the system throughput can be obtained.
\end{proof}
\end{theorem}

\section{Simulation Results}
In this section, we conduct numerical simulations to verify the effectiveness of the proposed scheme. In the simulation, the specific Rayleigh fading is considered with channel gain expectations $\Omega_{1}=E\{g_{1}\}$ and $\Omega_{2}=E\{g_{2}\}$.

In the Fig. 3, we plot the probabilities of different instantaneous CSI regions versus the fixed transmission rate $R_{0}$ under different system parameters. In particular, in the subfigure (a), $\sigma_{r}^2=\sigma_{d}^2=1$, $I_{r}=5$, $\Omega_{1}=0.8$ and $\Omega_{2}=0.6$, $P_{1}=P_{2}=25$ dB. In this case, we note that although $P_{1}=P_{2}$, the $S$-$R$ link is stronger than the $R$-$D$ link, which leads to the statistic CSI cases $\Psi_{1}$, $\Psi_{2}$ and $\Psi_{3}$. In the subfigure (b), we set the transmit power $P_{1}=23.75$ dB and $P_{2}=25$ dB. Although $\Omega_{1}>\Omega_{2}$,  but $P_{1}<P_{2}$, which offsets the disparity of the channel variances. Due to the symmetry, there is only one case $\Psi_{3}$. In the subfigure (c), we set $P_{1}=25$ dB and $P_{2}=30$ dB. In this case, although $P_{1}<<P_{2}$, which results to the cases $\Psi_{3}$, $\Psi_{4}$ and $\Psi_{5}$. In this figure, we reveal that the statistic cases from $\Psi_{1}$ to $\Psi_{5}$ exist under different system setups.
\begin{figure}[t]
\centering
\includegraphics[width=3.2in]{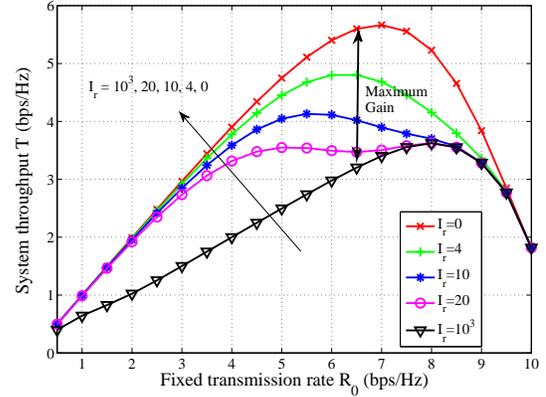}
\vspace{-2mm}
\caption{System throughput vs. the fixed transmission rate $R_{0}$}
\end{figure}

In Fig. 4, we plot the system throughput versus the fixed transmission rate $R_{0}$. In particular, the system parameters are set as follows: $P_{1}=P_{2}=30$ dB, $\sigma_{d}^2=\sigma_{r}^2=1$, $\Omega_{1}=0.8$, $\Omega_{2}=0.6$. We note that the fixed transmission rate $R_{0}$ remarkably affects the system throughput. In the low and strong RSI cases, the system throughput will firstly increase with $R_{0}$, after reaching the maximum value, the system throughput will decrease with the continuing increase of $R_{0}$. Thus, proper design the value of $R_{0}$ can optimize the system performance. However, in the case $I_{r}=20$, there are two values that both maximize the system throughput, which indicates the joint effects of the RSI and fixed rate $R_{0}$ on the system performance. In addition,  in the high $R_{0}$ regime, non-RSI achieves the same performance compared to the strong RSI. The reason is that all the time slots in the region $\mathcal{R}_{1}$ and $\mathcal{R}_{2}$ are allocated to the mode $M_{2}$. However, the sum probability of the region $\mathcal{R}_{1}$ and $\mathcal{R}_{2}$ has no relationship with the variance of the RSI.

\begin{figure}[t]
\centering
\includegraphics[width=3.2in]{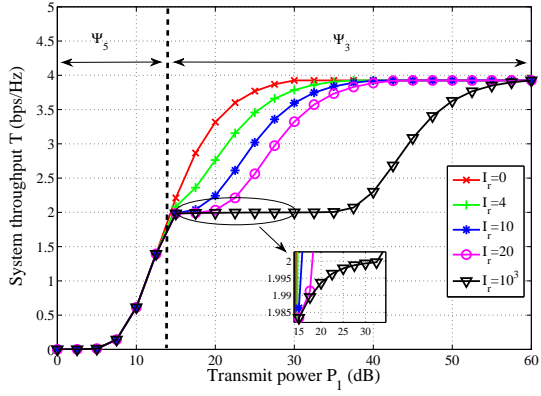}
\vspace{-2mm}
\caption{System throughput vs. the transmit power $P_{1}$.}
\end{figure}

Fig. 5 plots the system throughput versus the transmit power $P_{1}$. We set $P_{2}=30$ dB, $\Omega_{1}=\Omega_{2}=0.8$, $\sigma_{r}^2=\sigma_{d}^2=1$, $R_{0}=4$ bps/Hz. We note that the whole figure can be categorized into two cases: $\Psi_{5}$ in the low transmit power regime and $\Psi_{3}$ in the high transmit power regime. In the case $\Psi_{5}$, the $S$-$R$ link is strengthened with the increase of $P_{1}$ and the system throughput increases significantly. Due to $P_{\mathcal{R}_{1}}$ is very small, the system will almost directly convert to the case $\Psi_{3}$. In the case $\Psi_{3}$, the system throughput first increases very slowly, and then fastly. The reason is that, the increase of $P_{1}$ mainly leads to the increase of $P_{\mathcal{R}_{2}}$. However, the time slots in the region $\mathcal{R}_{2}$ need to be allocated to the modes $M_{1}$ or $M_{2}$. Thus, the system throughput increases slowly. In the high transmit power regime, for example, $P_{1} \in (40, 60)$ dB, the increase of $P_{1}$ mainly leads to the increase of $P_{\mathcal{R}_{1}}$. In the region $R_{1}$, all the time slots are allocated to the mode $M_{3}$. Thus, the system throughput increases rapidly.

Fig. 6 plots the system throughput versus the transmit power $P_{2}$. We note that the system throughput keeps the same in the low $P_{2}$ regime under different RSI cases. The reason is that in the low $P_{2}$ regime, the system performance is restricted by the $R-D$ link, which is none relative to the RSI. With the increase of $P_{2}$, the system throughput will increase proportionally. However, it is observed that for the strong RSI case, i.e., $I_{r}=P_{2}$, the system throughput reach the maximum value when $P_{2}\approx 18$ dB. Then, the throughput will decrease with the increase of $P_{2}$ and finally converge to a constant. The reason is that the RSI increases with the transmit power, which reduces the SINR performance at the relay node.

\begin{figure}[t]
\centering
\includegraphics[width=3.2in]{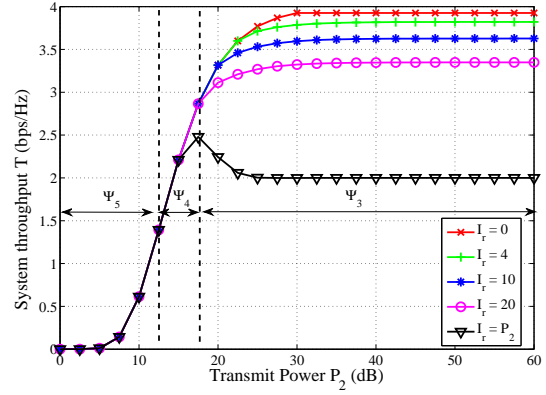}
\vspace{-2mm}
\caption{System Throughput vs. the transmit power $P_{2}$.}
\end{figure}
\section{Conclusion}
In this paper, we designed an optimal hybrid FD/HD scheme for the one-way DF FD relay system with a buffer. We focused on the throughput maximization by the proposed scheme. To solve this problem, we divided the system into four different transmission modes and formulated the original problem as a binary integer problem. By relaxing the problem, we obtained the optimal mode selection probabilities based on the KKT optimal conditions. The proposed scheme not only achieved the optimal FD mode or HD mode selection, but also realized the adaptive $S$-$R$ or $R$-$D$ link selection. Simulation results validated that the proposed scheme could fully exploit the potential of the FD mode.

\appendix
We note that the relaxed problem $\mathcal{P}_{2}$ is a linear programming problem. Since the Karush-Kuhn-Tucker (KKT) conditions are necessary for the optimal solution, we first examine the KKT conditions of the problem $\mathcal{P}_{2}$. The Lagrangian function for the problem $\mathcal{P}_{2}$ is given by
\begin{align}\label{757457}
  &\mathcal{L}(d_{j}(i),\alpha,\beta(i), \mu_{j}(i), \nu_{j}(i))=\notag\\
   &\ -\frac{1}{N} \sum_{i=1}^{N}[d_{2}(i)O_{2}(i)+d_{3}(i)O_{3}(i)]\times R_{0}  \notag\\
   &\ +\ \frac{\alpha_{0}}{N} \sum_{i=1}^{N}[d_{2}(i)O_{2}(i)-d_{1}(i)O_{1}(i)]\times R_{0} \notag\\
   &\ +\beta(i)(\sum_{i=1}^{N}\sum_{j=1}^{4}d_{j}(i)-1)+\sum_{i=1}^{N}\sum_{j=1}^{4}\mu_{j}(i)(d_{j}(i)-1)\notag\\
   &\ -\sum_{i=1}^{N}\nu_{j}(i)d_{j}(i)
\end{align}
where $\alpha_{0},\beta(i),\mu_{j}(i), \nu_{j}(i)$ are the non-negative Lagrange multipliers of the constraints from $\text{C1}$ to $\text{C4}$. For the optimal solution, the derivatives of (\ref{757457}) equal to zero, i.e.,
\begin{equation}\label{472389}
  \frac{\partial \mathcal{L}}{\partial d_{j}(i)} = 0, \forall\ i, j
\end{equation}
for $j=1,2,3,4$, we have
\begin{align}\label{5743978}
  V_{1}(i)=&\ N[\beta(i)+\mu_{1}(i)-\nu_{1}(i)] = \alpha_{0}O_{1}(i)R_{0} \notag\\
  V_{2}(i)=&\ N[\beta(i)+\mu_{2}(i)-\nu_{2}(i)] = (1-\alpha_{0})O_{2}(i)R_{0}\notag\\
  V_{3}(i)=&\ N[\beta(i)+\mu_{3}(i)-\nu_{3}(i)] = O_{3}(i)R_{0} \notag\\
  V_{4}(i)=&\ N[\beta(i)+\mu_{4}(i)-\nu_{4}(i)] = 0
\end{align}
where $V_{j}(i)$ is defined as the selection function. For simplicity, we first examine that the mode $M_{1}$ is selected. In this case, $d_{1}(i)=1$, we get $d_{j}(i)=0, j\in\{2,3,4\}$, thus, the constraints C3 for $j=2, 3,4$ and C4 for $j=1$ are inactive. According to the complementary slackness condition, we get that $\mu_{j}(i)=0, j=2,3,4$  and $\nu_{1}(i)=0$. Here, we find that
\begin{equation}\label{56496}
  V_{1}(i)-V_{j}(i)=N[\mu_{1}(i)+\nu_{j}(i)],\ j=2,3,4
\end{equation}
Since $\mu_{1}(i)\ge 0$ and $\nu_{j}(i)\ge0$, $V_{1}(i)-V_{j}(i)\ge 0, j=2,3,4$. This implies that when $M_{1}$ is selected, $V_{1}(i)\ge V_{j}(i), \forall j, j=2,3,4$. The results will be similar if we set $d_{j}(i)=1,j=2,3,4$. Now we can conclude that the mode with the largest selection function should be selected in each time slot.

Based on the property of the selection function, next we will discuss the different values of $\alpha_{0}$ in the right hand side of (\ref{5743978}). We consider the region $\mathcal{R}_{1}$, in this region, all of $O_{1}(i)$, $O_{2}(i)$ and $O_{3}(i)$ equal to 1. According to the KKT conditions, we know that $\alpha_{0}\ge 0$. First, we consider that $\alpha_{0}=0$, $V_{2}(i)=V_{3}(i)>V_{1}(i)$, indicating mode $M_{2}$ or mode $M_{3}$ should be selected. We assume that the selection probability of the mode $M_{j}$ is $P_{1}^j$, thus the constraint C1 leads to the following equation
\begin{align}\label{623943}
  &P_{\mathcal{R}_{4}}+P_{\mathcal{R}_{5}}= P_{\mathcal{R}_{1}}P_{1}^2+P_{\mathcal{R}_{2}}+P_{\mathcal{R}_{3}}\notag\\
  &P_{1}^2+P_{1}^3=1\notag\\
  \Rightarrow\quad& P_{1}^2=\frac{P_{\mathcal{R}_{4}}+P_{\mathcal{R}_{5}}-P_{\mathcal{R}_{2}}-P_{\mathcal{R}_{3}}}{P_{\mathcal{R}_{1}}}
\end{align}
where $P_{\mathcal{R}_{2}}$ is based on the preliminary 2, i.e., the time slots in the region $\mathcal{R}_{2}$ are preferred to be utilized to balance the constraint C1. The selection of the mode $M_{2}$ in the region $\mathcal{R}_{1}$ means all the time slots in the region $\mathcal{R}_{2}$ are allocated to the mode $M_{2}$. From (\ref{623943}), we note that if $P_{1}^2<1$, we have $P_{\mathcal{R}_{4}}+P_{\mathcal{R}_{5}}-P_{\mathcal{R}_{2}}-P_{\mathcal{R}{1}}< P_{\mathcal{R}_{3}}\leq P_{\mathcal{R}_{4}}+P_{\mathcal{R}_{5}}-P_{\mathcal{R}_{2}}$, i.e., $\Psi_{2}$.

However, if $P_{1}^2=1$, we get $P_{\mathcal{R}_{3}}\leq P_{\mathcal{R}_{4}}+P_{\mathcal{R}_{5}}-P_{\mathcal{R}_{2}}-P_{\mathcal{R}_{1}}$, i.e., $\Psi_{1}$. In this case, even if all the time slots in the region $\mathcal{R}_{1}$ are allocated to the mode $M_{2}$, $R_{1} \geq R_{2}$. Hence, we need to select the silent mode $M_{4}$ in the region $\mathcal{R}_{4}+\mathcal{R}_{5}$ to balance the arrival rate and departure rate. We have the following equation
\begin{align}\label{459545985}
  &(P_{\mathcal{R}_{4}}+P_{\mathcal{R}_{5}})P_{4}^{1}= P_{\mathcal{R}_{1}}+P_{\mathcal{R}_{2}}+P_{\mathcal{R}_{3}}\notag\\
  \Rightarrow\quad& P_{4}^1=\frac{P_{\mathcal{R}_{1}}+P_{\mathcal{R}_{2}}+P_{\mathcal{R}_{3}}}{P_{\mathcal{R}_{4}}+P_{\mathcal{R}_{5}}},
\end{align}
since that there are only two modes $M_{1}$ and $M_{4}$ available in the region $\mathcal{R}_{4}$ and $\mathcal{R}_{5}$, the selection probability of the mode $M_{4}$ is $P_{4}^4=1-P_{4}^1$.

Next, if $\alpha_{0}\in(0,1)$, $V_{3}(i)>\text{max}\{V_{1}(i),V_{2}(i)\}$, thus, the mode $M_{3}$ will be selected with the probability $1$ in the region $\mathcal{R}_{1}$. This case implies that by the adjustment of the allocation of the time slots in the region $\mathcal{R}_{2}$ to the mode $M_{1}$ and $M_{2}$ can realize the equality of the constraint C1. Hence, we have $P_{\mathcal{R}_{4}}+P_{\mathcal{R}_{5}}-P_{\mathcal{R}_{2}}< P_{\mathcal{R}_{3}}\leq P_{\mathcal{R}_{4}}+P_{\mathcal{R}_{5}}+P_{\mathcal{R}_{2}}$, i.e., $\Psi_{3}$. For the region $\mathcal{R}_{2}$, we have
\begin{align} \label{78474849}
&P_{\mathcal{R}_{4}}+P_{\mathcal{R}_{5}}+P_{\mathcal{R}_{2}}P_{2}^{1}= P_{\mathcal{R}_{2}}P_{2}^{2}\notag\\
&P_{2}^{1}+P_{2}^{2}=1\notag\\
\Rightarrow\quad & P_{2}^1=\frac{P_{\mathcal{R}_{2}}+P_{\mathcal{R}_{4}}+P_{\mathcal{R}_{5}}}{2P_{\mathcal{R}_{3}}}, \notag\\
 \quad&P_{2}^2 = 1-P_{2}^1,
\end{align}

If $\alpha_{0}=1$, $V_{1}(i)=V_{3}(i)>V_{2}(i)$, mode $M_{1}$ and $M_{3}$ should be selected in the region $\mathcal{R}_{1}$. This case corresponds to that $P_{\mathcal{R}_{2}}+P_{\mathcal{R}_{4}}+P_{\mathcal{R}_{5}}< P_{\mathcal{R}_{3}}\leq P_{\mathcal{R}_{5}}+P_{\mathcal{R}_{4}}+P_{\mathcal{R}_{2}}+P_{\mathcal{R}_{1}}$, i.e., $\Psi_{4}$. The selection probability can be expressed as the following
\begin{align}\label{14641156}
  &P_{\mathcal{R}_{4}}+P_{\mathcal{R}_{5}}+P_{\mathcal{R}_{2}}+P_{\mathcal{R}_{1}}P_{1}^1= P_{\mathcal{R}_{3}}\notag\\
  &P_{1}^1+P_{1}^3=1\notag\\
  \Rightarrow\quad& P_{1}^1=\frac{P_{\mathcal{R}_{3}}-P_{\mathcal{R}_{2}}-P_{\mathcal{R}_{4}}-P_{\mathcal{R}_{5}}}{P_{\mathcal{R}_{1}}}\notag\\
  &P_{1}^3=1-P_{1}^1,
\end{align}

In this case, the mode $M_{1}$ is selected in the region $\mathcal{R}_{1}$, which means that all the time slots in the region $\mathcal{R}_{2}$ are allocated to the mode $M_{1}$, i.e., $P_{2}^1=1$ in the region $\mathcal{R}_{2}$.

Last, we consider that $\alpha_{0}>1$, we note that in this case, $V_{1}(i)=\text{max}\{V_{j}(i)\}$, $j=1,2,3,4$, meaning that only the mode $M_{1}$ is selected in the region $\mathcal{R}_{1}$, which corresponds to the case $P_{\mathcal{R}_{3}}> P_{\mathcal{R}_{1}}+P_{\mathcal{R}_{2}}+P_{\mathcal{R}_{4}}+P_{\mathcal{R}_{5}}$, i.e., case $\Psi_{5}$. Hence, to satisfy the equality in the constraint C1, the inactive mode $M_{4}$ should be selected in the time slots of the region $\mathcal{R}_{3}$. The selection probability is given by
\begin{align}\label{46845913}
   & P_{\mathcal{R}_{1}}+P_{\mathcal{R}_{2}}+ P_{\mathcal{R}_{4}}+ P_{\mathcal{R}_{5}}=P_{\mathcal{R}_{3}}P_{1}^2\notag \\
    \Rightarrow \quad & P_{3}^2=\frac{P_{\mathcal{R}_{1}}+P_{\mathcal{R}_{2}}+P_{\mathcal{R}_{4}}+P_{\mathcal{R}_{5}}}{P_{\mathcal{R}_{3}}},
\end{align}
and $P_{3}^4=1-P_{3}^2$. Hitherto, all the cases from $\Psi_{1}$ and $\Psi_{5}$ have been discussed. Based on the property of the KKT conditions, the optimal solution of the hybrid FD/HD mode selection problem is obtained. The proof is completed.

\bibliographystyle{IEEEtran}
\bibliography{reference}
\end{document}